\documentclass[letter,11pt]{article}
\usepackage{hyperref}
\usepackage[left=1in,right=1in,top=1in,bottom=1in]{geometry}
\usepackage{amsmath,amssymb,amsthm}

\newtheorem{theorem}{Theorem}[section]
\newtheorem{corollary}[theorem]{Corollary}
\newtheorem{question}[theorem]{Question}
\newtheorem*{lemma*}{Lemma}

\newcommand{\ket}[1]{| #1 \rangle}
\newcommand{\bra}[1]{\langle #1 |}
\newcommand{\T}{\text{MM}}
\newcommand{\id}{\mathbf{1}}
\newcommand{\tr}{\mathop{\mathrm{tr}}}
\newcommand{\inner}[2]{\left\langle #1 \!\mid\! #2 \right\rangle}

\newcommand{\C}{\mathbb{C}}
\newcommand{\R}{\mathbb{R}}

\newcommand{\Z}{\mathbb{Z}}
\DeclareMathOperator{\GL}{GL}

\title{Designing Strassen's Algorithm}
\author{Joshua A. Grochow\footnote{Departments of Computer Science and Mathematics, University of Colorado, Boulder, CO \texttt{jgrochow@colorado.edu}} \,and Cristopher Moore\footnote{The Santa Fe Institute, Santa Fe, NM \texttt{moore@santafe.edu}}}

\begin{document}
\maketitle

\begin{abstract}
In 1969, Strassen shocked the world by showing that two $n \times n$ matrices could be multiplied in time asymptotically less than $O(n^3)$. While the recursive construction in his algorithm is very clear, the key gain was made by showing that $2 \times 2$ matrix multiplication could be performed with only 7 multiplications instead of 8. The latter construction was arrived at by a process of elimination and appears to come out of thin air. Here, we give the simplest and most transparent proof of Strassen's algorithm that we are aware of, using only a simple unitary 2-design and a few easy lines of calculation. %from the representation theory of finite groups and a few easy lines of calculation.
Moreover, using basic facts from the representation theory of finite groups, we use 2-designs coming from group orbits to generalize our construction to all $n \geq 2$ (although the resulting algorithms aren't optimal for $n \ge 3$). 
\end{abstract}

\section{Introduction}
The complexity of matrix multiplication is a central question in computational complexity, bearing on the complexity not only of most problems in linear algebra, but also of myriad combinatorial problems, e.g., various shortest path problems~\cite{zwick} and bipartite matching problems~\cite{sankowski}. The main question around matrix multiplication is whether two $n \times n$ matrices can be multiplied in time $O(n^{2 + \varepsilon})$ for every $\varepsilon > 0$. The current best upper bound on this exponent is $2.3728639$~\cite{LeGall}, narrowly beating~\cite{stothers,VW}. The best known lower bound is still only $3n^2 - o(n)$~\cite{landsbergLB}. 

Since Strassen's 1969 paper~\cite{strassen}, which showed how to beat the standard $O(n^3)$ time algorithm, it has been understood that one way to get asymptotic improvements in algorithms for matrix multiplication is to find algebraic algorithms for multiplying small matrices using only a few multiplications, and then to apply these algorithms recursively. 

While the recursive construction in Strassen's algorithm is very clear---treat a $2n \times 2n$ matrix as a $2 \times 2$ matrix each of whose entries is an $n \times n$ matrix---the base case, which accounts for how Strassen was able to beat $O(n^3)$, seems to come out of thin air. Indeed, Strassen was trying to prove, by process of (intelligently exhaustive) elimination, that such an algorithm could \emph{not} exist (e.g.,~\cite[Remark~1.1.1]{landsbergSurvey} or~\cite{liptonRegan}). In his paper it is presented as follows, which ``one easily sees''~\cite[p.~355]{strassen} correctly computes $2 \times 2$ matrix multiplication $C = A B$:
\[
\begin{array}{rclcrcl}
C_{11} & = & I + IV - V + VII && C_{12} & = & III + V \\
C_{21} & = & II + IV && C_{22} & = & I + III - II + VI
\end{array} \, ,
\]
where
\[
\begin{array}{rclcrcl}
I & = & (A_{11} + A_{22})(B_{11} + B_{22}) && V & = & (A_{11} + A_{12})B_{22} \\
II & = & (A_{21} + A_{22}) B_{11} && VI & = & (-A_{11} + A_{21})(B_{11} + B_{12})\\
III & = & A_{11} (B_{12} - B_{22}) && VII & = & (A_{12} - A_{22})(B_{21} + B_{22}) \\
IV & = & A_{22} (-B_{11} + B_{21}) 
\end{array} \, .
\]
While verifying the above by calculation is not difficult---after all, it's only seven multiplications and four linear combinations---it is rather un-illuminating.  In particular, the verification gives no sense of \emph{why} such a decomposition exists.

In this paper, we give a proof of Strassen's algorithm that is the most transparent we are aware of.  The basic idea is to note that Strassen's algorithm has a symmetric group of vectors lurking in it, which form what is known as a (unitary) 2-design.  Using the representation theory of finite groups, we obtain generalizations to higher dimensions, which suggest further directions to explore in our hunt for efficient algorithms.

\subsection{Other explanations of Strassen's algorithm}
Landsberg~\cite[Section~3.8]{landsbergSurvey} points out that Strassen's algorithm could have been anticipated because the border-rank of any $4 \times 4 \times 4$ tensor is at most seven. Although this may lead one to suspect the existence of an algorithm such as Strassen's, it does not give an explanation for the fact that the rank (rather than border-rank) of $2 \times 2$ matrix multiplication is at most seven, nor does it give an explanation of Strassen's particular construction.

Several authors have tried to make Strassen's construction more transparent through various calculations, e.\,g., \cite{gastinel, yuval, chatelin, alekseyev, paterson, GK, minz, CILO}.
While these lend some insight, and some provide proofs that are perhaps easier to remember (and teach) than Strassen's original presentation, each of them either involves some ad hoc constructions or some un-illuminating calculations, which are often left to the reader. We feel that they do not really offer \emph{conceptual} explanations for the fact that the rank of $\T_2$ is at most 7.

Clausen~\cite{clausen} (see~\cite[pp.~11--12]{BCS} for a more widely available explanation in English) showed how one can use group orbits to show that the rank of $\T_2$ is at most 7. In fact, Clausen's beautiful construction was one of the starting points of our investigation. However, that construction relies on a seemingly magical property of a certain $4 \times 4$ multiplication table. %, and is not immediately obvious how to generalize to all $n$. 
More recently, Ikenmeyer and Lysikov \cite{ikenmeyer-lysikov} gave a beautiful explanation of Clausen's construction, but ultimately their proof for Strassen's algorithm still relies on the same magical property of the same $4 \times 4$ multiplication table, and it is not immediately obvious how to generalize to all $n$. 
In contrast, our result easily generalizes to all $n$, and more generally to orbits of any irreducible representation of any finite group.

\subsection{Related work}
This paper is a simplified and self-contained version of Section 5 of~\cite{GM}, in which we explored highly symmetric algorithms for multiplying matrices. Recently, there have been several papers analyzing the geometry and symmetries of algebraic algorithms for small matrices~\cite{burichenkoStrassen, burichenko, landsbergRyder, landsbergMichalekLB1,CILO}. In~\cite{GM}, we tried to take this line of research one step further by using symmetries to discover \emph{new} algorithms for multiplying matrices of small size. While those algorithms did not improve the state-of-the-art bounds on the matrix multiplication exponent, they suggested that we can use group symmetries and group orbits to find new algorithms for matrix multiplication.  In addition to their potential value for future endeavors, we believe that these highly symmetric matrix multiplication algorithms are beautiful in their own right, and deserve to be shared simply for their beauty.  

Although the method of construction suggested in~\cite{GM}, and independently in~\cite{CILO}, is more general than this, the constructions we ended up finding in~\cite{GM} were in fact all instances of a single design-based construction yielding $n^3 - n + 1$ multiplications for $n \times n$ matrix multiplication. The proof that this construction works is the simplest and most transparent proof of Strassen's algorithm that we are aware of.  

One may also reasonably wonder whether there is any relationship between our group-based construction and the family of group-based constructions suggested by Cohn and Umans~\cite{CU03}, including the constructions given in~\cite{CKSU05} and generalizations in~\cite{CU13}. While there may be a common generalization that captures both methods, at the moment we don't know of any direct relationship between the two. Indeed, one cannot use the group-theoretic approach of~\cite{CU03} to explain Strassen's result, even though the constructions of~\cite{CKSU05} get a better exponent: The only way to use their approach for the $2 \times 2$ case is to embed $\T_2$ into the cyclic group $C_7$, but Cohn and Umans showed that one could not beat $n^3$ using only abelian groups.  (Some of their more complicated constructions can beat $n^3$ in abelian groups, but those involve embedding multiple copies of $\T$ into the same group simultaneously, whereas here we are explicitly talking about embedding a single copy of $\T_2$.)

\section{Complexity, symmetry, and designs}

For general background on algebraic complexity, we refer the reader to the book~\cite{BCS}. Bl\"{a}ser's survey article~\cite{blaserSurvey}, in addition to excellent coverage around matrix multiplication, has a nice tutorial on tensors and their basic properties. 

In the algebraic setting, since matrix multiplication is a bilinear map, it is known that it can be reformulated as a tensor, and that the algebraic complexity of matrix multiplication is within a factor of 2 of the rank of this tensor. The matrix multiplication tensor for $n \times n$ matrices is 
\begin{equation}
\label{eq:twisted}
\T^{abc}_{def} = \delta^a_e \delta^b_f \delta^c_d \, , 
\end{equation}
where the indices range from $1$ to $n$, and where $\delta$ is the Kronecker delta, $\delta^a_b = 1$ if $a=b$ and $0$ if $a \ne b$.  This is also defined by the inner product
\[
\inner{\T}{A \otimes B \otimes C} = \tr ABC \, .
\]

Given vector spaces $V_1, \dotsc, V_k$, a vector $v \in V_1 \otimes \dotsb \otimes V_k$ is said to have tensor rank one if it is a separable tensor, that is, of the form $v_1 \otimes v_2 \otimes \dotsb \otimes v_k$ for some $v_i \in V_i$. The \emph{tensor rank} of $v$ is the smallest number of rank-one tensors whose sum is $v$. In the case of $\T_n$, we have $k = 3$ and $V_1 \cong V_2 \cong V_3 \cong \C^{n \times n}$. 

The matrix multiplication tensor $\T$ is characterized by its symmetries (e.g.,~\cite{burgisserIkenmeyer}).  That is, up to a constant, it is the unique operator fixed under the following action of $\GL(n)^3$: given $X, Y, Z \in \GL(n)$, we have
\begin{equation}
\label{eq:action}
\T = (X \otimes Y \otimes Z) \T (Z^{-1} \otimes X^{-1} \otimes Y^{-1}) \, , 
\end{equation}
where, if the notation isn't already clear, it will become so in the next equation.  To see that $\T$ has this symmetry, note that 
\[
\tr ABC = \tr (Z^{-1} A X)(X^{-1} B Y)(Y^{-1} C Z) \, .
\]
The fact that $\T$ is the only such operator up to a constant comes from a simple representation-theoretic argument, which generalizes the fact that the only matrices which are invariant under conjugation are scalar multiples of the identity.

This suggests that a good way to search for matrix multiplication algorithms is to start with sums of separable tensors where the sum has some symmetry built in from the beginning.  As we will see, one useful kind of symmetry is the following.  We say that a set of $n$-dimensional vectors $S$ is a \emph{unitary 2-design} if it has the following two properties:
\begin{equation}
\label{eq:design}
\sum_{v \in S} v = 0 
\quad \text{and} \quad
\frac{1}{|S|} \sum_{v \in S} \ket{v} \bra{v} = \frac{1}{n} \,\id \, ,
\end{equation}
where $\id$ denotes the identity matrix.  Here we use the Dirac notation $\ket{u}\bra{v}$ for the outer product of $u$ and $v$, i.e., the matrix whose $i,j$ entry is $u_i v_j^*$ where $*$ denotes the complex conjugate.  

%To see that $\T$ is the only such operator up to a constant, we simply note that if $\rho \cong \C^n$ is the standard representation of $\GL(n)$, then $\rho^{\otimes 3}$ is an irreducible representation of $\GL(n)^3$.  Taking the tensor product with its dual $(\rho^{\otimes 3})^*$ then gives a single copy of the trivial representation, so $\T$ lives in---and therefore spans---this one-dimensional subspace. 
%%(Note, however, that if we think of $U_1,U_2,U_3 = \rho \otimes \rho^*$, with $\T$ a 3-tensor in $U_1 \otimes U_2 \otimes U_3$, then the action of $\GL(n)^3$ is \emph{not} the simplest natural embedding into $\GL(U_1) \times \GL(U_2) \times \GL(U_3) \cong \GL(n^2)^3$, where each copy of $\GL(n)$ gets mapped into exactly one of the $\GL(U_i)$. Rather, the action is given by shifting the indices: $(X,Y,Z) \in \GL(n)^3$ maps to $(X \otimes Z^{-1}, X^{-1} \otimes Y, Y^{-1} \otimes Z) \in \GL(U_1) \times \GL(U_2) \times \GL(U_3)$.)
%
%The same argument applies to any group $G$ that has an $n$-dimensional irreducible presentation, which we will also call $\rho$.  Since $\rho^{\otimes 3}$ is an irrep of $G^3$, $\rho^{\otimes 3} \otimes (\rho^{\otimes 3})^*$ contains a unique one-dimensional subspace of operators invariant under the action of $G^3$ where $X, Y, Z$ are each $\rho(g)$ for some $g \in G$. 

The following theorem shows how 2-designs can be used to construct matrix multiplication algorithms.
\begin{theorem} 
\label{thm:main}
Let $S \subset \C^n$ be a unitary 2-design, and let $s = |S|$.   Then the tensor rank of $\T_n$ is at most $s(s-1)(s-2)+1$.
\end{theorem}
\noindent 

\begin{proof}
Let $S=\{w_1, \dotsc, w_s \}$.  
We will show that the following is a decomposition of $\T_n$:
\begin{equation}
\label{eq:lattice}
\T_n = \id^{\otimes 3} 
+ \frac{n^3}{s^3} \sum_{\mbox{$i,j,k$ distinct}} 
\ket{w_i}\bra{w_j-w_i} \otimes
\ket{w_j}\bra{w_k-w_j} \otimes
\ket{w_k}\bra{w_i-w_k} \, .
\end{equation}
Since there are $s(s-1)(s-2)$ distinct ordered triples $i,j,k \in \{1,\ldots,s\}$, this decomposition has $s(s-1)(s-2)+1$ terms.

To prove~\eqref{eq:lattice}, we use the fact~\eqref{eq:twisted} that $\T$ can be written as a kind of twisted tensor product of identity matrices or Kronecker deltas.  By the second property in the definition~\eqref{eq:design} of a 2-design, we have
\begin{equation}
\label{eq:twisted-sum}
\T_n = \frac{n^3}{s^3} \sum_{i,j,k}
\ket{w_i}\bra{w_j} \otimes
\ket{w_j}\bra{w_k} \otimes
\ket{w_k}\bra{w_i} 
\end{equation}
At the same time, the un-twisted version of this identity is
\begin{equation}
\label{eq:id-sum}
\id^{\otimes 3} = \frac{n^3}{s^3} \sum_{i,j,k}
\ket{w_i} \bra{w_i} \otimes
\ket{w_j}\bra{w_j} \otimes
\ket{w_k}\bra{w_k} \, . 
\end{equation}
Now we expand~\eqref{eq:lattice}.  We can sum over all $s^3$ triples $i,j,k$, since if any of these are equal the summand is zero.  
Then
\begin{align} 
\sum_{i,j,k} &
\ket{w_i}\bra{w_j-w_i} \otimes
\ket{w_j}\bra{w_k-w_j} \otimes
\ket{w_k}\bra{w_i-w_k} \nonumber \\
&= 
\sum_{i,j,k}
\ket{w_i}\bra{w_j} \otimes
\ket{w_j}\bra{w_k} \otimes
\ket{w_k}\bra{w_i} \nonumber \\
&- \sum_{i,j,k} \Big[ 
\ket{w_i}\bra{w_i} \otimes
\ket{w_j}\bra{w_k} \otimes
\ket{w_k}\bra{w_i} 
+ \ket{w_i}\bra{w_j} \otimes
\ket{w_j}\bra{w_j} \otimes
\ket{w_k}\bra{w_i} 
+ \ket{w_i}\bra{w_j} \otimes
\ket{w_j}\bra{w_k} \otimes
\ket{w_k}\bra{w_k} \Big] \label{eq:line3} \\
&+ \sum_{i,j,k} \Big[ 
\ket{w_i}\bra{w_j} \otimes
\ket{w_j}\bra{w_j} \otimes
\ket{w_k}\bra{w_k} 
+ \ket{w_i}\bra{w_i} \otimes
\ket{w_j}\bra{w_k} \otimes
\ket{w_k}\bra{w_k} 
+ \ket{w_i}\bra{w_i} \otimes
\ket{w_j}\bra{w_j} \otimes
\ket{w_k}\bra{w_i} \Big] \label{eq:line4} \\
&- \sum_{i,j,k}
\ket{w_i}\bra{w_i} \otimes
\ket{w_j}\bra{w_j} \otimes
\ket{w_k}\bra{w_k} 
\nonumber 
\end{align}
The mixed terms in lines~\eqref{eq:line3} and~\eqref{eq:line4} disappear because each product has an index that appears only once, and the first property in the definition~\eqref{eq:design} of a 2-design implies that summing over that index gives zero.  Combining this with~\eqref{eq:twisted-sum} and~\eqref{eq:id-sum} leaves us with 
\[
\frac{n^3}{s^3} \sum_{i,j,k} 
\ket{w_i}\bra{w_j-w_i} \otimes
\ket{w_j}\bra{w_k-w_j} \otimes
\ket{w_k}\bra{w_i-w_k} 
= \T_n - \id^{\otimes 3} \, , 
\]
which completes the proof.
\end{proof}

Now, in $n=2$ dimensions the three corners of an equilateral triangle form a 2-design:
\begin{equation}
\label{eq:equil}
S = \big\{ (1,0) \, , \; (-1/2, \sqrt{3}/2) \, , \; (-1/2, -\sqrt{3}/2) \big\} \, . 
\end{equation}
The outer products of these vectors with themselves are 
\[
\begin{pmatrix} 1 & 0 \\ 0 & 0 \end{pmatrix} \, , \;
\begin{pmatrix} 1/4 & -\sqrt{3}/4 \\ -\sqrt{3}/4 & 3/4 \end{pmatrix} \, , \;
\begin{pmatrix} 1/4 & \sqrt{3}/4 \\ \sqrt{3}/4 & 3/4 \end{pmatrix} \, , \;
\]
and the average of these is $\id / 2$.  This design~\eqref{eq:equil} has size $s=3$, in which case Theorem~\ref{thm:main} shows that $\T_2$ has a tensor rank of at most $7$.  

The reader might object that we haven't really re-derived Strassen's algorithm since our algorithm doesn't seem to yield the same equations as Strassen's. However, de Groote~\cite{deGroote} has shown that all 7-term decompositions of $\T_2$ are equivalent up to a change of basis, i.e., an instance of the $\GL(n)^3$ action~\eqref{eq:action}.  Thus the algorithm based on the triangular design~\eqref{eq:equil} is in fact isomorphic to Strassen's algorithm, and in any case, it gives a conceptual explanation for the fact that $\T_2$ has tensor rank $7$.  (For the reader wondering about algorithms for matrix multiplication over rings other than $\C$, see Section~\ref{sec:future}.)

\section{Generalizations to larger $n$ from group orbits}

The triangular design~\eqref{eq:equil} has a pleasing symmetry.   In this section we show how to find similar designs in higher dimensions as the orbits of group actions.  We assume basic familiarity with finite groups, for which we refer the reader to any standard textbook such as~\cite{artin}.  We need a few facts from representation theory, which we spell out for completeness in the hopes of making the paper self-contained for a larger audience.  Everything we do will be over the complex numbers $\C$, but generalizes to other fields, with some modifications.

A \emph{representation} of a finite group $G$ is a vector space $V$ together with a group homomorphism $\rho\colon G \to \GL(V)$, where $\GL(V)$ denotes the general linear group of $V$, namely, the group of all invertible linear transformations from $V$ to itself. By choosing a basis for $V$, we identify $V \cong \C^{\dim V}$, and each $\rho(g)$ becomes a $\dim V \times \dim V$ matrix such that $\rho(g)\rho(h) = \rho(gh)$ for all $g,h \in G$.  

When the homomorphism $\rho$ is understood from context, we refer to $V$ as a representation of $G$. In this case, for $g \in G$ and $v \in V$ we write $gv$ instead of $\rho(g)(v)$.  The trivial representation is the identity map on $V$, where $gv=v$ for all $g \in G$.  

A representation $V$ of $G$ is called \emph{unitary} if each $\rho(g)$ is a unitary matrix. Any representation of a finite group over $\C$ is equivalent, up to change of basis, to a unitary representation.  
%; in fact, on each irreducible representation of $G$ there is a $G$-invariant inner product $\inner{u}{v} = \inner{gu}{gv}$ which is unique up to scalar multiples, and a corresponding norm $|v|^2 = \inner{v}{v} = |gv|^2 > 0$ for all $v \ne 0$.  
In this basis we define the inner product $\inner{u}{v} = \sum_i u^*_i v_i$ and the norm $|v|^2 = \inner{v}{v} = \sum_i |v_i|^2$ as usual.  Note that $\inner{u}{v} = \inner{gu}{gv}$ and $|gv|^2 = |v|^2$.  

Given two representations $\rho\colon G \to \GL(V)$ and $\rho'\colon G \to \GL(V')$ of the same group $G$, their \emph{direct sum} is a representation $(\rho \oplus \rho')\colon G \to \GL(V \oplus V')$ given by the matrices
\[
(\rho \oplus \rho')(g) = \left(\begin{array}{cc}
\rho(g) & 0 \\
0 & \rho'(g)
\end{array}\right).
\]
A representation $\rho\colon G \to \GL(V)$ is \emph{irreducible} if the only subspaces $W \subseteq V$ that are sent to themselves by every $g \in G$, i.e., such that $\rho(g)(w) \in W$ for all $g \in G$ and $w \in W$, are the trivial subspaces $W=0$ or $W=V$.  In fields of characteristic zero such as $\C$ or $\R$, a representation is a direct sum if and only if it is not irreducible.  

Given two representations $V,W$ of $G$, a \emph{homomorphism} of representations is a linear map $\varphi \colon V \to W$ that commutes with the action of $G$, in the sense that $\varphi(g v) = g \varphi(v)$ for all $g \in G, v \in V$.  

\begin{lemma*}[Schur's Lemma]
If $V$ and $W$ are two irreducible representations of a group $G$, then every nonzero homomorphism $V \to W$ is invertible. In particular, over an algebraically closed field, every homomorphism $V \to V$ is a scalar multiple of the identity $\id_V$.
\end{lemma*}
Schur's Lemma implies that the orbit of any unit-length vector in an irreducible representation is a 2-design in the sense defined above.  We include a proof of this classical fact for completeness.

\begin{corollary} 
\label{cor:schur1}
If $V$ is a nontrivial irreducible representation of $G$, and $v \in V$ with $|v|^2=1$, then the orbit $\{ gv \mid g \in G \}$ is a 2-design.
\end{corollary}

\begin{proof}
First, the vector $\sum_{g \in G} gv$ always spans a 1-dimensional trivial sub-representation $W \subseteq V$, since $h \cdot (\sum gv) = \sum_{g \in G} (hg) v = \sum_{g \in G} gv$. If $V$ is irreducible, then either $W=0$ or $W=V$, but if $V$ is nontrivial we cannot have $W=V$.  Thus $\sum_g gv = 0$.  

Second, let $\varphi = \sum_{g \in G} \ket{gv}\bra{gv}$.  Then for any $h \in G$ we have
\[ 
h \varphi(w) = h \left( \sum_g \ket{gv}\bra{gv} \right) w = \sum_g \ket{hgv}\inner{gv}{w} = \sum_g \ket{gv} \inner{h^{-1} gv}{w} = \sum_g \ket{gv} \inner{gv}{hw} = \varphi(hw) \, .
\] 
Thus $\varphi$ is a homomorphism of representations, and by Schur's Lemma $\varphi$ is a multiple of $\id_V$.  We obtain the scaling factor by taking traces:
\[
\frac{1}{|G|} \tr \varphi = \frac{1}{|G|} \sum_g |gv|^2 = |v|^2 = 1 = \frac{1}{\dim V} \,\tr \id_V \, , 
\]
so $\varphi / |G| = \id_V / \dim V$.
\end{proof}

We can now combine Corollary~\ref{cor:schur1} with Theorem~\ref{thm:main} to produce matrix multiplication algorithms in all dimensions, by considering families of finite groups and their irreducible representations.  In particular, we have the following.

\begin{corollary} 
\label{cor:main}
For every $n \geq 1$, the tensor rank of $\T_n$ is at most $n^3 - n + 1$. 
\end{corollary}

\begin{proof}
Let $G = S_{n+1}$ be the symmetric group acting by permuting the coordinates on $V = \R^{n+1}$.  As a representation of $G$, this splits into a direct sum of the trivial representation (spanned by the all-ones vector) and the so-called ``standard representation'' $V$ of $S_{n+1}$, of dimension $n$ (consisting of the vectors whose coordinates sum to zero). Let $w \in V \subset \R^{n+1}$ be the unit vector 
\[
w = \frac{1}{\sqrt{n(n+1)}} \,( n, -1, \ldots, -1) 
\]
The orbit of $w$ has size $s=n+1$, and consists of unit vectors pointing to the corners of a simplex.  Now apply Theorem~\ref{thm:main}.
\end{proof}

\section{Future directions} \label{sec:future}

\paragraph{Highly symmetric algorithms?} Since any design must span $\C^n$, it has size $s \ge n$.  Indeed, since its elements sum to zero, they are linearly dependent, so $s \ge n+1$.  Thus the simplex designs of Corollary~\ref{cor:main} are optimal in this context, and applying Theorem~\ref{thm:main} to larger $n$ cannot directly improve the matrix multiplication exponent.  This leaves open the question of whether there are other families of highly symmetric algorithms.  See~\cite{GM, CILO} for work in this direction, as well as~\cite{burichenkoStrassen, burichenko, landsbergRyder, landsbergMichalekLB1, ikenmeyer-lysikov}.

\paragraph{Using $t$-designs for $t > 2$?} The key fact we used was that any orbit in an irreducible representation of a finite group is a unitary 2-design.  Similarly, a  \emph{unitary $t$-design} in a vector space $V$ is a set of vectors $S \subseteq V$ such that, for every polynomial $f$ on $V$ of degree at most $t$, the average of $f$ over $S$ is the same as the average of $f$ over the unit sphere in $V$.  (Over the reals, these are traditionally called ``spherical $t$-designs,'' but we are working in complex vector spaces.)  Another open question is then
\begin{question}
Can $t$-designs for $t > 2$ help us construct efficient matrix multiplication algorithms?
\end{question}
\noindent For example, one might hope for a similar construction to Theorem~\ref{thm:main}, in which one could leverage the $t$-design property to get even more terms to cancel. 

\paragraph{Working over arbitrary rings?} 

One fact which is obvious from Strassen's original construction, but not from ours, is that the rank of $\T_2$ is at most 7 over any ring. Strassen's construction works in all rings since it only uses coefficients in $\{0,\pm 1\}$; ours uses coefficients in $\Z[\sqrt{3}, 1/2, 1/3]$, so it works in any ring where the elements $\sqrt{3}, 1/2, 1/3$ exist. Note that $1/2$ and $1/3$ exist in any ring $R$ of characteristic $m$ coprime to $6$, since any such ring contains $\Z/m\Z$ as a subring, in which $2$ and $3$ are units. If $\sqrt{3}$ doesn't exist in $R$, we can formally adjoin it by considering $R' = R[x]/(x^2-3)$; by a standard trick (e.g.~\cite[Section~15.3]{BCS}), this implies the same exponent over $R$ itself (although it may not actually yield an algorithm for $\T_2$ over $R$).

\begin{question}
Is there a similarly transparent and conceptual proof of Strassen's result that works over arbitrary rings?
\end{question}

\section*{Acknowledgments}

Parts of this project were inspired in 2015 by discussions with Jonah Blasiak, Thomas Church, Henry Cohn, and Chris Umans, via a collaboration funded by the AIM SQuaRE program, with an additional visit hosted by the Santa Fe Institute. J.A.G. was funded by an Omidyar Fellowship from the Santa Fe Institute during this work and by NSF grant DMS-1620484, and C.M. was funded partly by the John Templeton Foundation. C.M. also thanks \'{E}cole Normale Sup\'{e}rieure for providing a visiting position during which some of this work was carried out.

\bibliographystyle{alphaurl}
\bibliography{strassen-fourier}
\end{document}